\documentclass[lettersize,journal]{IEEEtran}
\usepackage{amsmath,amsfonts,amssymb,amsthm}
\usepackage{array}
\usepackage{textcomp}
\usepackage{stfloats}
\usepackage{url}
\usepackage{verbatim}
\usepackage{graphicx}
\def\BibTeX{{\rm B\kern-.05em{\sc i\kern-.025em b}\kern-.08em
    T\kern-.1667em\lower.7ex\hbox{E}\kern-.125emX}}
\usepackage{balance}
\usepackage{cite}
\usepackage{xcolor}

\usepackage{caption}
\usepackage{subcaption}
\usepackage[labelformat=simple]{subcaption}

\usepackage{pdfpages}

\usepackage{tikz}
\usetikzlibrary{tikzmark,decorations.pathreplacing}
\usepackage{rotating}

\newtheorem*{theorem*}{Theorem}

\newtheorem*{proof*}{Proof}
\newtheorem*{remark*}{Remark}
\newtheorem{remark}{Remark}

\usepackage{algorithm}
\usepackage[end]{algpseudocode}

\algnewcommand\algorithmicforeach{\textbf{for each}}
\algdef{S}[FOR]{ForEach}[1]{\algorithmicforeach\ #1\ \algorithmicdo}

\usepackage{fancyhdr}

\begin{document}
\title{Sparse Array Design for Dual-Function Radar-Communications System}  
\author{Huiping Huang, \textit{Student Member, IEEE}, Linlong Wu, \textit{Member, IEEE}, Bhavani Shankar, \textit{Senior Member, IEEE}, and Abdelhak M. Zoubir, \textit{Fellow, IEEE} \vspace{-8mm}
\thanks{The work of H. Huang was supported by the Graduate School CE within the Centre for Computational Engineering at Technische Universität Darmstadt. The work of L. Wu and B. Shankar was supported by FNR through the CORE SPRINGER project under Grant C18/IS/12734677. \emph{(Corresponding author: Linlong Wu)}}
\thanks{H. Huang was with Technische Universität Darmstadt, Darmstadt, Germany; and now is with Chalmers University of Technology, Gothenburg, Sweden (e-mail: huiping@chalmers.se).}
\thanks{L. Wu and B. Shankar are with University of Luxembourg, Luxembourg City, Luxembourg (e-mail: \{linlong.wu; Bhavani.Shankar\}@uni.lu).}
\thanks{A. M. Zoubir is with Technische Universität Darmstadt, Darmstadt, Germany (e-mail: zoubir@spg.tu-darmstadt.de).}
}

\markboth{}
{Sparse Array Design for Dual-function Radar-communications System}

\maketitle


\begin{abstract}
The problem of sparse array design for dual-function radar-communications is investigated. Our goal is to design a sparse array which can simultaneously shape desired beam responses and serve multiple downlink users with the required signal-to-interference-plus-noise ratio levels. Besides, we also take into account the limitation of the radiated power by each antenna. The problem is formulated as a quadratically constrained quadratic program with a joint-sparsity-promoting regularization, which is NP-hard. The resulting problem is solved by the consensus alternating direction method of multipliers, which enjoys parallel implementation. Numerical simulations exhibit the effectiveness
and superiority of the proposed method which leads to a more power-efficient solution.
\end{abstract}


\section{Introduction}
\label{introduction}
\IEEEPARstart{D}{ual}-function radar-communications (DFRC) systems have been recently widely investigated \cite{Mishra2019, Liu2020, Shi2022, Cheng2022}. They find applications in a wide range of areas, including vehicular networks, indoor positioning and covert communications \cite{Yang2015, Wymeersch2017}. To achieve accurate sensing and high throughput, the DFRC systems probably require a large number of antennas \cite{Chen2017, Elbir2022, Raei2022, Kassaw2019}. In practice, a base station usually equips with less radio-frequency (RF) chains than antennas given the hardware cost consideration. For such a configuration, it raises a question on how to adaptively switch the available RF chains to the corresponding subset of antennas \cite{Hamza2021, Wei2021, Huang2022Sep}, which can be interpreted as sparse array design.    

In this work, we consider the sparse array design for the DFRC system. Similar problems have been studied in \cite{Wang2022, Xu2022, Zhang2022, Wang2018, Wang2019, Ahmed2020}. The authors in \cite{Wang2022} derived a Cram\'er-Rao bound for the cooperative radar‐communications system, where they focused on target parameter estimation. The work \cite{Xu2022} developed an antenna selection strategy by using a learning approach. This method requires a training process, which might be unavailable in some practical applications. The authors in \cite{Zhang2022} introduced a realistic waveform constraint in the DFRC system, in order to improve the power efficiency. A surrogate subproblem instead of the original problem was solved, which might lead to a highly suboptimal solution. In \cite{Wang2018} and \cite{Wang2019}, several types of DFRC systems were proposed which implement simultaneous beamformers associated with single and different sparse arrays with shared aperture. However, these two works consider a single-user case, and the corresponding methods cannot be applicable to the case of multi-users. A weighted $\ell_{1,q}$-norm optimization based approach was developed in \cite{Ahmed2020}. Note that all the problems considered in \cite{Wang2018, Wang2019, Ahmed2020} are convex and were solved by the existing toolbox such as CVX.

Unlike in previous works, we propose a novel system model for DFRC systems, and formulate the corresponding sparse array design problem as a quadratically constrained quadratic program (QCQP) regularized by a joint-sparsity-promoting term. Besides the control on illumination beampattern and communication signal-to-interference-plus-noise ratio (SINR), we also take into account the limitation of the radiated power by each antenna. We propose an algorithm based on the consensus alternating direction method of multipliers (ADMM) to solve the resulting problem. Note that at each ADMM iteration, the primary variable has a closed-form solution, and the auxiliary variables can be solved efficiently in a parallel manner. Simulation results show its superior performance compared to other examined methods.

{\color{black}It is worth mentioning that the sparse array beamforming can be incorporated into the framework of
hybrid beamforming, via decomposing the sparse array beamformer into the baseband beamformer with a particular selection preference and the RF beamformer. Therefore, it can further improve the power efficiency within the hybrid beamforming.
}

The remainder of the paper is organized as follows. The system model is established in Section \ref{SignalModel}. The proposed algorithm is presented in Section \ref{ProposedMethod_Section}. Simulation results are shown in Section \ref{simulation}, while Section \ref{conclusion} concludes the paper.

\textit{Notations}: Throughout this paper, bold-faced lower-case (upper-case) letters denote vectors (matrices). Superscripts $\cdot^{\mathrm{T}}$, $\cdot^{\mathrm{H}}$, $\cdot^{*}$, $\cdot^{-1}$ denote transpose, Hermitian transpose, conjugate, and inverse, respectively. $\mathrm{E}\{\cdot\}$ denotes the expectation operator. $|\cdot|$ and $\angle \cdot$ are the modulus and phase, respectively, both in an element-wise manner. $\mathbb{C}$ and $\mathbb{R}$ are the sets of complex and real numbers, respectively, and $\jmath = \sqrt{-1}$. $\|\cdot\|_{0}$, $\|\cdot\|_{1}$, and $\|\cdot\|_{2}$ are $\ell_{0}$-quasi-norm, $\ell_{1}$-norm, and $\ell_{2}$-norm, respectively. $\otimes$ denotes the Kronecker product. $\oslash$ is the element-wise division operator. ${\bf I}_{N}$ and ${\bf O}_{M \times N}$ are the $N \times N$ identity matrix and $M \times N$ zero matrix, respectively. ${\bf 1}$ and ${\bf 0}$ are the all-ones and all-zeros vector of appropriate size, respectively.  

\section{System Model and Problem Formulation}  
\label{SignalModel}
We consider a DFRC system, as indicated in Fig. \ref{systemmodel_fig}. The transmitter, consisting of $N$ antennas, emits $M$ signals, each to a single user. The weight vector ${\bf w}_{m} \in \mathbb{C}^{N}$ is designed to transfer the data symbol, $s_{m}(t)$, to user $m$, $\forall m = 1, 2, \cdots \!, M$. All the $M$ transmitted signals are also received by a target. The total transmitted signal can then be formulated as   \vspace{-1mm}
\begin{align}
\label{transmit_signal}
{\bf x}(t)=\sum_{m=1}^{M}{\bf w}_{m}s_{m}(t),
\end{align}
where $t$ is the time index. We assume that the data symbols, $s_{m}(t)$ $\forall m = 1, 2, \cdots \!, M$, are mutually uncorrelated, and each $s_{m}(t)$ is zero-mean, spatially white with unit variance, i.e.,  \vspace{-1mm}
\begin{subequations}
\label{signal_property}
\begin{align}
\mathrm{E}\{|s_{m}(t)|^{2}\} = & ~\! 1, ~ \forall m \!=\! 1, 2, \cdots \!, M,  \\
\mathrm{E}\{|s_{i}(t) s_{j}^{*}(t)|\} = & ~\! 0, ~ \forall i, j \!=\! 1, 2, \cdots \!, M, \text{~\!and~} i \!\neq\! j.
\end{align} 
\end{subequations}

\begin{figure}[t]
	\centerline{\includegraphics[width=0.42\textwidth]{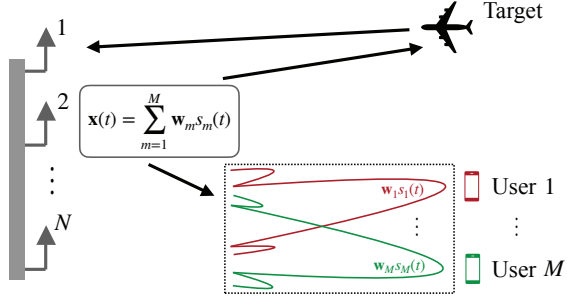}}  
	\caption{Illustration of a DFRC system.}
	\label{systemmodel_fig}  \vspace{-4mm}
\end{figure}

{\color{black}For the radar system, we wish to focus the transmit energy on a certain angular region, say $\mathcal{B}$, in order to have a higher probability of detecting the target(s) lying therein. To this end, we expect that the array response within $\mathcal{B}$ is not less than a preset threshold, while the response outside $\mathcal{B}$ (denoted as $\bar{\mathcal{B}}$) is not higher than a small threshold.} Mathematically,
\begin{subequations}
\label{pass_stop_band}
\begin{align}
\label{passband}
\mathrm{E} \! \left\{ {\bf x}^{\mathrm{H}}(t){\bf a}(\theta_{i}){\bf a}^{\mathrm{H}}(\theta_{i}){\bf x}(t) \right\} & \geq \epsilon_{{\mathrm{p}}}, ~ \forall \theta_{i} \in \mathcal{B}, \\
\label{stopband}
\mathrm{E} \! \left\{ {\bf x}^{\mathrm{H}}(t){\bf a}(\theta_{j}){\bf a}^{\mathrm{H}}(\theta_{j}){\bf x}(t) \right\} & \leq \epsilon_{{\mathrm{s}}}, ~ \forall \theta_{j} \in \bar{\mathcal{B}},
\end{align}
\end{subequations}
where ${\bf a}(\theta)$ denotes the steering vector of ${\theta}$, $\epsilon_{{\mathrm{p}}}$ and $\epsilon_{{\mathrm{s}}}$ are two preset thresholds corresponding to the mainlobe and sidelobe, respectively. Substituting (\ref{transmit_signal}) and (\ref{signal_property}) into (\ref{pass_stop_band}), we have  \vspace{-1mm}
\begin{subequations}
\label{pass_stop_band_w}
\begin{align}
\label{passband_w}
\sum_{m = 1}^{M} {\bf w}_{m}^{\mathrm{H}}{\bf a}(\theta_{i}){\bf a}^{\mathrm{H}}(\theta_{i}){\bf w}_{m} & \geq \epsilon_{{\mathrm{p}}}, ~ \forall \theta_{i} \in \mathcal{B}, \\
\label{stopband_w}
\sum_{m = 1}^{M} {\bf w}_{m}^{\mathrm{H}}{\bf a}(\theta_{j}){\bf a}^{\mathrm{H}}(\theta_{j}){\bf w}_{m} & \leq \epsilon_{{\mathrm{s}}}, ~ \forall \theta_{j} \in \bar{\mathcal{B}}.
\end{align}
\end{subequations}

For the downlink communication systems, the received data by the $m$-th user is given as   \vspace{-2mm}
\begin{align}
y_{m}(t) = \underbrace{ {\bf h}_{m}^{\mathrm{H}}{\bf w}_{m}s_{m}(t) }_{\text{communication signal}} + \underbrace{ \sum_{\substack{j=1, j \neq m}}^{M} \!\! {\bf h}_{m}^{\mathrm{H}}{\bf w}_{j}s_{j}(t)}_{\text{interference signals}} + \underbrace{n_{m}(t)}_{\text{noise}},
\end{align}
where $n_{m}(t)$ is the additive white Gaussian noise with mean zero and known variance $\sigma_{m}^{2}$, $\forall m = 1, 2, \cdots \!, M$. Besides, ${\bf h}_{m} \in \mathbb{C}^{N}$ denotes the channel state information (CSI) vector, which models the propagation loss and phase shift of the frequency-flat quasi-static channel from the transmitter to the $m$-th user. The CSI, $\{{\bf h}_{m}\}_{m}$, is assumed to be perfectly available at the transmitter. The maximum power radiated by each antenna $n$ is given by $P_{n}$. Therefore, we have \cite{Chen2017}  \vspace{-1mm}
\begin{align}
\label{power_antenna_constraint}
\sum_{m = 1}^{M} {\bf w}_{m}^{\mathrm{H}}{\bf E}_{n}{\bf w}_{m} \leq P_{n}, ~ \forall n = 1, 2, \cdots \!, N,
\end{align}
where ${\bf E}_{n}$ is the $N \times N$ all-zero matrix except the $n$-th diagonal entry being $1$. The received SINR of user $m$ is defined below, and a certain minimum SINR needs to be guaranteed, i.e.,  \vspace{-1mm}
\begin{align}
\label{SINR_constraint}
\text{SINR}_{m} \! \triangleq \! \frac{ \left| {\bf h}_{m}^{\mathrm{H}} {\bf w}_{m} \right|^{2}}{  \sum_{j \neq m} \left| {\bf h}_{m}^{\mathrm{H}} {\bf w}_{j} \right|^{2} \! + \! \sigma_{m}^{2} } \! \geq \! \gamma_{m}, ~ \forall m = 1, \cdots \!, M,
\end{align}
where $ \sum_{j \neq m}$ is short notation for $ \sum_{j = 1, j \neq m}^{M}$.

Now we consider the scenario where only $K \leq N$ RF chains are available, and thus only $K$ antennas can be simultaneously utilized in the DFRC system. Therefore, ${\bf w}_{m}$, $\forall m$, are sparse, and moreover, they share the same sparsity pattern. To fulfill this requirement, we need:  \vspace{-1mm}
\begin{align}
\label{sparsity_w}
\|{\widetilde{\bf w}}\|_{0} \leq K,
\end{align}
where ${\widetilde{\bf w}} \triangleq [{\widetilde{w}}_{1}, {\widetilde{w}}_{2}, \cdots \!, {\widetilde{w}}_{N}]^{\mathrm{T}}$ with its component defined as ${\widetilde{w}}_{n} \triangleq \|[{\bf w}_{1}(n), {\bf w}_{2}(n), \cdots \!, {\bf w}_{M}(n)]\|_{2}$ $\forall n = 1, 2, \cdots, N$, and ${\bf w}_{m}(n)$ being the $n$-th entry of vector ${\bf w}_{m}$, $\forall m = 1, 2, \cdots \!, M$.

The quality of service (QoS) problem aims at minimizing the total transmit power (TxPower), with several constraints described above. To this end, the QoS problem is given as   \vspace{-1mm}
\begin{subequations}
\label{prob_L0}
\begin{align}
\label{def_TxPower}
 \min_{ \{{\bf w}_{m}\} } ~ & {\mathrm{TxPower}} \triangleq \sum_{m = 1}^{M} \|{\bf w}_{m}\|_{2}^{2} \\
 \mathrm{s.t.} ~~~\! & (\ref{pass_stop_band_w}), (\ref{power_antenna_constraint}), (\ref{SINR_constraint}), \text{~\!and~} (\ref{sparsity_w}).
\end{align} 
\end{subequations}
By replacing the non-convex $\ell_{0}$-quasi-norm in (\ref{sparsity_w}) with the $\ell_{1}$-norm, and incorporating it into Problem (\ref{prob_L0}) as a regularization term, we have   \vspace{-1mm}
\begin{align}
\label{prob_L1}
 \min_{ \{{\bf w}_{m}\} } ~\!\! \sum_{m = 1}^{M} \!\! \|{\bf w}_{m}\|_{2}^{2} + \eta\|{\widetilde{\bf w}}\|_{1} \quad \mathrm{s.t.} ~ (\ref{pass_stop_band_w}), (\ref{power_antenna_constraint}), \text{~\!and~} (\ref{SINR_constraint}),
\end{align} 
where $\eta$ is a parameter controlling the sparsity of the solution.

\section{Proposed Method}
\label{ProposedMethod_Section}
The difficulties of solving Problem (\ref{prob_L1}) are twofold. First, the discorporated sparsity has a joint structure among all ${\bf w}_{m}$ \cite{Huang2022}, which is different from the classical sparsity and consequently, the existing solutions cannot be directly extended. Second, the constraints (\ref{passband_w}) and (\ref{SINR_constraint}) are nonconvex.

In what follows, we first deal with the joint-sparsity structure of the solution by reformulating the problem, and then, we develop an algorithm based on the consensus ADMM.  \vspace{-2mm}

\subsection{Problem Reformulation}
First of all, we define a long column vector ${\bf w} \in \mathbb{C}^{M \! N}$ as ${\bf w} \triangleq [{\bf w}_{1}^{\mathrm{T}}, {\bf w}_{2}^{\mathrm{T}}, \cdots, {\bf w}_{M}^{\mathrm{T}}]^{\mathrm{T}} $ and define $M$ matrices ${\bf \Phi}_{m} \in \mathbb{R}^{N \times M \! N}$ as ${\bf \Phi}_{m} \triangleq [{\bf O}_{N \times (m-1)N}, {\bf I}_{N}, {\bf O}_{N \times ({M}-m)N}]$. Then, Problem (\ref{prob_L1}) can be reformulated as  \vspace{-1mm}
\begin{subequations}
\label{prob_L21}
\begin{align}
\min_{{\bf w}} ~ & \|{\bf w}\|_{2}^{2} + \eta\|{\bf w}\|_{2,1} \\
\label{w_theta_i}
\mathrm{s.t.} ~~\! & {\bf w}^{\mathrm{H}} {\bf A}_{\theta_{i}} {\bf w} \geq \epsilon_{{\mathrm{p}}}, ~ \forall \theta_{i} \in \mathcal{B}, \\
\label{w_theta_j}
& {\bf w}^{\mathrm{H}} {\bf A}_{\theta_{j}} {\bf w} \leq \epsilon_{{\mathrm{s}}}, ~ \forall \theta_{j} \in \bar{\mathcal{B}}, \\
\label{w_B}
& {\bf w}^{\mathrm{H}} {\bf B}_{n} {\bf w} \leq P_{n}, ~ \forall n = 1, 2, \cdots \!, N, \\
\label{SINR_C}
& \frac{{\bf w}^{\mathrm{H}} {\bf C}_{m} {\bf w} }{ {\bf w}^{\mathrm{H}} {\bf C}_{\bar{m}} {\bf w} + \sigma_{m}^{2} } \geq \gamma_{m}, ~ \forall m = 1, 2, \cdots \!, M,
\end{align}
\end{subequations}
where we define {\color{black}${\bf A}_{\theta_{i}} \triangleq \sum_{m = 1}^{M} \! {\bf \Phi}_{m}^{\mathrm{H}}{\bf a}(\theta_{i}){\bf a}^{\mathrm{H}}(\theta_{i}){\bf \Phi}_{m}$, ${\bf A}_{\theta_{j}} \triangleq \sum_{m = 1}^{M} \! {\bf \Phi}_{m}^{\mathrm{H}}{\bf a}(\theta_{j}){\bf a}^{\mathrm{H}}(\theta_{j}){\bf \Phi}_{m}$, ${\bf B}_{n} \triangleq \sum_{m = 1}^{M} \! {\bf \Phi}_{m}^{\mathrm{H}}{\bf E}_{n}{\bf \Phi}_{m}$, ${\bf C}_{m} \triangleq {\bf \Phi}_{m}^{\mathrm{H}}{\bf h}_{m}{\bf h}_{m}^{\mathrm{H}}{\bf \Phi}_{m}$, ${\bf C}_{\bar{m}} \triangleq \sum_{j \neq m} \!\! {\bf \Phi}_{j}^{\mathrm{H}}{\bf h}_{m}{\bf h}_{m}^{\mathrm{H}}{\bf \Phi}_{j}$,} and we have utilized $ \sum_{m = 1}^{M} \|{\bf w}_{m}\|_{2}^{2} = \|{\bf w}\|_{2}^{2}$, ${\bf w}_{m} = {\bf \Phi}_{m}{\bf w}$, and $\|{\bf w}\|_{2,1} \triangleq \sum_{n = 1}^{N} \sqrt{\sum_{m = 1}^{M}|{\bf w}_{m}(n)|^{2} }$. Further, since the denominators, ${\bf w}^{\mathrm{H}} {\bf C}_{\bar{m}} {\bf w} + \sigma_{m}^{2}$, are always positive, (\ref{SINR_C}) can be written as  \vspace{-1mm}
\begin{align}
\label{w_D}
{\bf w}^{\mathrm{H}}{\bf D}_{m}{\bf w} \geq \gamma_{m}\sigma_{m}^{2}, ~ \forall m = 1, 2, \cdots \!, M,
\end{align}
where ${\bf D}_{m} \triangleq {\bf C}_{m} - \gamma_{m}{\bf C}_{\bar{m}}$,  $\forall m = 1, 2, \cdots \!, M$. 
We observe that (\ref{w_theta_i}), (\ref{w_theta_j}), (\ref{w_B}), and (\ref{w_D}) take the form: ${\bf w}^{\mathrm{H}}{\bf F}{\bf w} \leq f$. Therefore, Problem (\ref{prob_L21}) can be reformulated as
\begin{subequations}
\label{prob_L21_F}
\begin{align}
\min_{{\bf w}} ~ & \|{\bf w}\|_{2}^{2} + \eta\|{\bf w}\|_{2,1} \\
\mathrm{s.t.} ~~\! & {\bf w}^{\mathrm{H}}{\bf F}_{l}{\bf w} \leq f_{l}, ~ \forall l = 1, 2, \cdots \!, L,
\end{align}
\end{subequations}
where the subscript $\cdot_{l}$ is used to indicate the $l$-th constraint in Problem (\ref{prob_L21}), and $L$ is the total number of constraints. It is worth noting that ${\bf F}_{l}$ corresponding to $-{\bf A}_{\theta_{i}}$ is negative semidefinite, and ${\bf F}_{l}$ corresponding to $-{\bf D}_{m}$ could be indefinite. Therefore, in general, Problem (\ref{prob_L21_F}) is non-convex and thus NP-hard \cite{Huang2016}. To solve this problem, we propose an algorithm based on the consensus ADMM, which is capable of handling all the constraints in parallel.

\subsection{Proposed Algorithm}
\label{proposedmethod}

Firstly, we formulate Problem (\ref{prob_L21_F}) by introducing $L$ auxiliary variables $\{{\bf v}_{l} \in \mathbb{C}^{M \! N} \}_{l = 1}^{L}$, and settle the original variable ${\bf w}$ and the auxiliary variables $\{{\bf v}_{l}\}_{l}$ in a separable fashion, as \vspace{-5mm}
\begin{subequations}
\label{QoS_joint_L21_F_wv_problem}
\begin{align}
 \min_{{\bf w}, \{{\bf v}_{l}\}} & ~ \|{\bf w}\|_{2}^{2} + \frac{\eta}{L} \sum_{i = 1}^{L} \|{\bf v}_{l}\|_{2,1} \\
\label{constraint_L21_F}
 \mathrm{s.t.} ~~\! & ~ {\bf v}_{l}^{\mathrm{H}}{\bf F}_{l}{\bf v}_{l} \leq {f}_{l}, ~ {\bf v}_{l} = {\bf w}, ~ \forall l =1,2, \cdots \!, L.
\end{align}
\end{subequations}

Next, we form the scaled-form augmented Lagrangian function related to the above problem, as: $ \mathcal{L}({\bf w}, \{{\bf v}_{l}\}, \{{\bf u}_{l}\}) = \|{\bf w}\|_{2}^{2} + \frac{\eta}{L} \! \sum_{l = 1}^{L} \! \|{\bf v}_{l}\|_{2,1} + \frac{\rho}{2} \! \sum_{l = 1}^{L} \! \left( \|{\bf v}_{l} \! - \! {\bf w} + {\bf u}_{l}\|_{2}^{2} - \|{\bf u}_{l}\|_{2}^{2} \right) $, where $\rho > 0$ stands for the augmented Lagrangian parameter, and ${\bf u}_{l}$ is the scaled dual variable corresponding to the equality constraint ${\bf v}_{l} = {\bf w}$ in Problem (\ref{QoS_joint_L21_F_wv_problem}).

Finally, the consensus ADMM updating equations can be written down as
\begin{subequations}
\label{admm_prob}
\begin{align}
\label{admm_w}
{\bf w} & \gets {\displaystyle{\arg\min_{{\bf w}}}} ~ \|{\bf w}\|_{2}^{2} + \frac{\rho}{2} \sum_{l = 1}^{L} \| {\bf v}_{l} - {\bf w} + {\bf u}_{l} \|_{2}^{2}, \\
\label{admm_v}
{\bf v}_{l} & \gets \left\{ \!\!\!
\begin{array}{l}
{\displaystyle{\arg\min_{{\bf v}_{l}}}} ~ \frac{\eta}{L}\|{\bf v}_{l}\|_{2,1} + \frac{\rho}{2}\|{\bf v}_{l} - {\bf w} + {\bf u}_{l}\|_{2}^{2} \\
~~~ \mathrm{s.t.} ~~~~ {\bf v}_{l}^{\mathrm{H}}{\bf F}_{l}{\bf v}_{l} \leq f_{l},
\end{array} \right. \\
\label{admm_u}
{\bf u}_{l} & \gets {\bf u}_{l} + {\bf v}_{l} - {\bf w}.
\end{align}
\end{subequations}

In what follows, we show how to solve ${\bf w}$ and ${\bf v}_{l}$ from (\ref{admm_w}) and (\ref{admm_v}), respectively. We start by solving ${\bf w}$ from (\ref{admm_w}). It is straightforward to see that, by calculating the derivative of the objective function of (\ref{admm_w}) with respect to (w.r.t.) ${\bf w}$ and setting it to ${\bf 0}$, we obtain the solution to (\ref{admm_w}) as
\begin{align}
{\bf{\widehat{ w}}} = \frac{\rho}{2 + \rho L}\sum_{l = 1}^{L}({\bf v}_{l} + {\bf u}_{l}).
\end{align}
On the other hand, to solve ${\bf v}_{l}$ from (\ref{admm_v}), we firstly consider the unconstrained minimization problem:
\begin{align}
\label{unconstrained_prob}
	\displaystyle\min_{{\bf v}_{l}} ~ f({\bf v}_{l}) \triangleq \frac{\eta}{L}\|{\bf v}_{l}\|_{2,1} + \frac{\rho}{2}\|{\bf v}_{l} - {\bf w} + {\bf u}_{l}\|_{2}^{2}.
\end{align}
The derivative of $f({\bf v}_{l})$ w.r.t. ${\bf v}_{l}$ is calculated as
\begin{align*}
\nabla_{ {\bf v}_{l} } f({\bf v}_{l}) = \left[ \frac{\eta}{L}\left( {\bf I}_{M} \otimes {\bf G} \right) + \rho {\bf I}_{M \! N} \right] \! {\bf v}_{l} - \rho({\bf w} - {\bf u}_{l}),
\end{align*}
where $ {\bf G} \! \in \! \mathbb{R}^{ N \! \times \! N}$ is a diagonal matrix with diagonal being \vspace{-2mm}
\begin{align*}
{\bf 1} \!\! \oslash \!\! \left[ \!\! \sqrt{ \!\! \sum_{m = 1}^{M} \!\! | \! {\bf v}_{l}(1 \! + \! (m \! - \! 1)N) \! |^{2}}, \cdots \! , \! \sqrt{ \!\! \sum_{m = 1}^{M} \!\! | \! {\bf v}_{l}(N \! + \! (m \! - \! 1)N) \! |^{2}} \right]^{\!\! \mathrm{T}} \!\! .
\end{align*} 
By setting the derivative to ${\bf 0}$, we obtain
\begin{align}
\label{v_l_KKT}
 \left[ \frac{\eta}{L}\left( {\bf I}_{M} \otimes {\bf G} \right) + \rho {\bf I}_{M \! N} \right] \! {\bf v}_{l} = \rho({\bf w} - {\bf u}_{l}).
\end{align}
Since $\left[ \frac{\eta}{L}\left( {\bf I}_{M} \! \otimes \! {\bf G} \right) \!+\! \rho {\bf I}_{M \! N} \right]$ is real-valued, $\angle{{\bf v}_{l}} = \angle{({\bf w} \!-\! {\bf u}_{l})}$. Thus, we only need to calculate the modulus of ${\bf v}_{l}$, using
\begin{align}
 \left[ \frac{\eta}{L}\left( {\bf I}_{M} \otimes {\bf G} \right) + \rho {\bf I}_{M \! N} \right] \! |{\bf v}_{l}| = \rho|{\bf w} - {\bf u}_{l}|.
\end{align}
By exploring the structure of ${\bf I}_{M} \! \otimes \! {\bf G}$, we observe that there are $N$ blocks each containing $M$ equal entries. Extracting the rows of the equal entries yields
\begin{align}
\label{v_l_n}
\left( \frac{\eta}{L\| {\bf v}_{l(n)} \|_{2}} + \rho \right) |{\bf v}_{l(n)}| = \rho |{\bf c}_{l(n)}|,
\end{align}
$\forall n = 1, 2, \cdots \!, N$, where ${\bf v}_{l(n)} \triangleq [{\bf v}_{l}(n), {\bf v}_{l}(n + N), \cdots \!, {\bf v}_{l}(n + (M-1)N)]^{\mathrm{T}} \in \mathbb{C}^{M}$, and ${\bf c}_{l(n)} \in \mathbb{C}^{M}$ contains the corresponding entries of $({\bf w} - {\bf u}_{l})$. From (\ref{v_l_n}), we have
\begin{align}
\label{v_l_n_abs}
 |{\bf v}_{l(n)}| = \frac{\rho L \|{\bf v}_{l(n)}\|_{2}}{\eta + \rho L \|{\bf v}_{l(n)}\|_{2}}|{\bf c}_{l(n)}|.
\end{align} 
Performing the element-wise square operation, we have
\begin{align}
 |{\bf v}_{l(n)} \! |^{2} = \frac{\rho^{2} L^{2} \|{\bf v}_{l(n)}\|_{2}^{2}}{\eta^{2} + \rho^{2}L^{2}\|{\bf v}_{l(n)}\|_{2}^{2} + 2 \eta \rho L \|{\bf v}_{l(n)}\|_{2}}|{\bf c}_{l(n)} \! |^{2}.
\end{align} 
Hence, we further have
\begin{subequations}
\begin{align}
\| \! {\bf v}_{l(n)} \! \|_{2}^{2} \! = & {\bf 1}^{\mathrm{T}}|{\bf v}_{l(n)} \! |^{2} \\
= & \frac{\rho^{2} L^{2} \| \! {\bf v}_{l(n)} \! \|_{2}^{2}}{\eta^{2} \!+\! \rho^{2}L^{2}\| \! {\bf v}_{l(n)} \! \|_{2}^{2} \!+\! 2 \eta \rho L \| \! {\bf v}_{l(n)} \! \|_{2}} {\bf 1}^{\mathrm{T}}|{\bf c}_{l(n)} \! |^{2}.
\end{align} 
\end{subequations}
The above equation leads to
\begin{align}
\label{sinpleQF}
\rho^{2}L^{2}\| \! {\bf v}_{l(n)} \! \|_{2}^{2} \!+\! 2 \eta \rho L \| \! {\bf v}_{l(n)} \! \|_{2} \!+\! \eta^{2} \!-\! \rho^{2}L^{2} {\bf 1}^{\mathrm{T}}|{\bf c}_{l(n)} \! |^{2} = 0,
\end{align}
the left-hand side of which is a simple quadratic function w.r.t. $\| \! {\bf v}_{l(n)} \! \|_{2}$, and its unique\footnote{Note that the quadratic function in (\ref{sinpleQF}) has two roots, one positive and one negative. In our case, the negative one is omitted, since its root $\| \! {\bf v}_{l(n)} \! \|_{2} \geq 0$.} root is given as
\begin{align}
\label{root_v_l_n}
\| {\bf v}_{l(n)} \|_{2} = \frac{\rho L \sqrt{{\bf 1}^{\mathrm{T}}|{\bf c}_{l(n)} \! |^{2} } - \eta }{\rho L}.
\end{align}

Then, by substituting (\ref{root_v_l_n}) into (\ref{v_l_n_abs}), we obtain
\begin{align}
\label{v_l_n_modulus}
|{\bf v}_{l(n)}| = \frac{\rho L \sqrt{{\bf 1}^{\mathrm{T}}|{\bf c}_{l(n)}|^{2}} - \eta }{ \rho L \sqrt{{\bf 1}^{\mathrm{T}}|{\bf c}_{l(n)}|^{2}} }|{\bf c}_{l(n)}|.
\end{align}
In (\ref{v_l_n_modulus}), we define $|{\bf v}_{l(n)}| = {\bf 0}$ if $|{\bf c}_{l(n)}| = {\bf 0}$. The solution for (\ref{v_l_KKT}), referred to as ${\bf{\bar v}}_{l}$, is finally obtained by combining ${\bf v}_{l(n)}$ (which can be calculated as $|{\bf v}_{l(n)}| e^{\jmath \angle{\bf v}_{l(n)}}$). Then, the solution to (\ref{admm_v}), denoted by ${\bf{\widehat v}}_{l}$, is found via the following theorem.

\begin{theorem*}
\label{v_closest}
If $\rho$ satisfies $\frac{\rho}{2} \gg \frac{\eta}{L}$, then ${\bf{\widehat v}}_{l}$ can be solved via:
\begin{align}
\label{v_l_hat}
{\bf{\widehat v}}_{l} \gets \arg\min_{{\bf v}_{l}} ~ \|{\bf v}_{l} - {\bf{\bar v}}_{l}\|_{2}^{2} \quad \mathrm{s.t.} ~ {\bf{ v}}_{l}^{\mathrm{H}}{\bf F}_{l}{\bf{ v}}_{l} \leq f_{l}.
\end{align}
\end{theorem*}

\begin{proof}
See Appendix \ref{proof_v_closest}.
\end{proof} 

\begin{remark}
Note that $\eta$ is related to the sparsity of the solution of Problem (\ref{prob_L1}), and $L$ is the total number of constraints in Problem (\ref{prob_L21_F}). Both of them are known for a specific problem. Hence, it is easy to choose a $\rho$ such that $\frac{\rho}{2} \gg \frac{\eta}{L}$.
\end{remark} 

\begin{remark}
If ${\bf{\bar v}}_{l}$ satisfies the constraint of (\ref{v_l_hat}), i.e., ${\bf{\bar v}}_{l}^{\mathrm{H}}{\bf F}_{l}{\bf{\bar v}}_{l} \leq f_{l}$, it is easy to have ${\bf{\widehat v}}_{l} = {\bf{\bar v}}_{l}$. Otherwise, notice that Problem (\ref{v_l_hat}) is a QCQP with one constraint, which can be solved optimally despite that ${\bf F}_{l}$ may be indefinite \cite{Huang2016}.
\end{remark}   

So far, we have presented how to solve ${\bf w}$ and ${\bf v}_{l}$ from (\ref{admm_w}) and (\ref{admm_v}), respectively. Note that $\{{\bf v}_{l}\}_{l = 1}^{L}$ and $\{{\bf u}_{l}\}_{l = 1}^{L}$ can be calculated in parallel. The complete consensus ADMM for solving Problem (\ref{prob_L21_F}) is summarized in Algorithm \ref{admm_alg}, in which $k_{\mathrm{max}}$ is used to terminate the iteration, and the superscript $\cdot^{(k)}$ denotes the corresponding variable at the $k$-th iteration. \vspace{-3mm}

\begin{algorithm}[t]
	\caption{Consensus ADMM for solving Problem (\ref{prob_L21_F})}
	\label{admm_alg}
	\textbf{Input:} $\eta$, $\rho$, $k_{\mathrm{max}}$, ${\bf F}_{l} \! \in \! \mathbb{C}^{M \! N \times M \! N}$\!, ${f}_{l}$, $\forall l = 1,2, \cdots \!, L$  \\
	\textbf{Output:} ${\bf{\widehat w}} \in \mathbb{C}^{M \! N}$ \\
	\textbf{Initialize:} ${\bf{\widehat{v}}}_{l}^{(0)} \! \gets \! {\bf v}_{l}^{\text{(init)}}$, ${\bf{\widehat{ u}}}_{l}^{(0)} \! \gets \! {\bf u}_{l}^{\text{(init)}}$, $k \! \gets \! 0$
	\begin{algorithmic}[1]
		\While {$k < k_{\mathrm{max}}$} 
		 \State ${\bf{\widehat{ w}}}^{(k + 1)} \gets \frac{\rho}{2 + \rho L} {\displaystyle{\sum_{l = 1}^{L}}} \! \left( {\bf{\widehat{v}}}_{l}^{(k)} + {\bf{\widehat{ u}}}_{l}^{(k)} \right)$
		 \ForEach {$l = 1, 2, \cdots \!, L$}  \tikzmarknode{v0}{}
		 \State ${\bf c}_{l} \gets {\bf{\widehat w}}^{(k + 1)} - {\bf{\widehat u}}_{l}^{(k)}$
		 \State $\angle {\bf{\bar v}}_{l}^{(k + 1)} \gets \angle {\bf c}_{l}$
		 	\ForEach {$n = 1, 2, \cdots \!, N$}   \tikzmarknode{v0_inner}{}
			\State $|{\bf{\bar v}}_{l(n)}^{(k + 1)}| \gets  \frac{\rho L \sqrt{{\bf 1}^{\mathrm{T}}|{\bf c}_{l(n)}|^{2}} - \eta }{ \rho L \sqrt{{\bf 1}^{\mathrm{T}}|{\bf c}_{l(n)}|^{2}} }|{\bf c}_{l(n)}|$  \!\!\!\! \tikzmarknode{dorthin_inner}{}
			\State ${\bf{\bar v}}_{l(n)}^{(k + 1)} \gets |{\bf{\bar v}}_{l(n)}^{(k + 1)}| e^{\jmath \angle{\bf{\bar v}}_{l(n)}^{(k + 1)}}$    \tikzmarknode{S_inner}{}
			\EndFor
		 \State Construct ${\bf{\bar v}}_{l}^{(k + 1)}$ using ${\bf{\bar v}}_{l(n)}^{(k + 1)}$, $\forall n$ 
		 \If {${\bf{\bar v}}_{l}^{(k + 1)\mathrm{H}}{\bf F}_{l}{\bf{\bar v}}_{l}^{(k + 1)} \! \leq \! f_{l}$} {${\bf{\widehat v}}_{l}^{(k + 1)} \!\! \gets \! {\bf{\bar v}}_{l}^{(k + 1)}$} ~\!\!\!\!\!\!\! \tikzmarknode{dorthin}{}
		 \Else $~{\bf{\widehat v}}_{l}^{(k + 1)} \gets \left\{ \!\! \begin{array}{l}
		 \displaystyle \arg\min_{{\bf v}_{l}} ~ \| {\bf v}_{l} - {\bf{\bar v}}_{l}^{(k + 1)}  \|_{2}^{2} \\
		  ~~~~ \mathrm{s.t.} \quad ~ {\bf{ v}}_{l}^{\mathrm{H}}{\bf F}_{l}{\bf{ v}}_{l} \leq f_{l} 
		  \end{array}
		  \right. $
		 \EndIf 
		 \State ${\bf{\widehat{ u}}}_{l}^{(k + 1)} \gets {\bf{\widehat{ u}}}_{l}^{(k)} + {\bf{\widehat v}}_{l}^{(k + 1)} - {\bf{\widehat w}}^{(k + 1)}$      \tikzmarknode{S}{}
		  \EndFor
		\State $k \gets k+1$
		\EndWhile
		\State ${\bf{\widehat w}} \gets {\bf{\widehat w}}^{(k)}$ 
	\end{algorithmic} 
\end{algorithm}      
\begin{tikzpicture}[overlay,remember picture]
 \draw[thick,decorate,decoration=brace] ([xshift=1mm,yshift=2mm]dorthin.east |- v0.south) 
 -- ([xshift=1mm,yshift=-2mm]dorthin.east |- S.south)
 node[midway,right = 1mm,align=left]{\begin{turn}{270}In Parallel \end{turn}};
\end{tikzpicture}
\begin{tikzpicture}[overlay,remember picture]
 \draw[thick,decorate,decoration=brace] ([xshift=1mm,yshift=2mm]dorthin_inner.east |- v0_inner.south) 
 -- ([xshift=1mm,yshift=-2mm]dorthin_inner.east |- S_inner.south)
 node[midway,right = 1mm,align=left]{\begin{turn}{270}In Parallel \end{turn}};
\end{tikzpicture}

\section{Simulation Results}
\label{simulation}
In this section, we evaluate the performance of the proposed algorithm compared with the feasible point pursuit successive convex approximation (FPP-SCA) method \cite{Mehanna2015}. FPP-SCA was proposed for general QCQP and it is adapted to our problem. Note that unlike our solution, no closed-form solution of FPP-SCA is given. Two metrics are adopted: the TxPower defined in (\ref{def_TxPower}) and the {\color{black}radar-side} mainlobe-to-sidelobe response ratio (MSRR) defined as 
$
\mathrm{MSRR} = \frac{  \sum_{\theta_{i} \in \mathcal{B}} \sum_{m = 1}^{M} \! \left| {\bf w}_{m}^{\mathrm{H}}{\bf a}(\theta_{i}) \right|^{2} }{ \sum_{\theta_{j} \in \bar{\mathcal{B}}} \sum_{m = 1}^{M} \! \left|{\bf w}_{m}^{\mathrm{H}}{\bf a}(\theta_{j}) \right|^{2}}
$.

We first consider a transmit system with a uniform linear array of $N = 10$ antennas and $K = 8$ or $10$ RF chains. {\color{black}There are $M = 2$ users located at $-45^{\circ}$ and $45^{\circ}$.} The mainlobe and sidelobe are $\mathcal{B} = [-5^{\circ} , 5^{\circ}]$ and {\color{black}$\bar{\mathcal{B}} = [-90^{\circ} , -60^{\circ}] \cup [-30^{\circ} , -20^{\circ}] \cup [20^{\circ} , 30^{\circ}] \cup [60^{\circ} , 90^{\circ}]$}. The thresholds for the mainlobe and sidelobe response are $\epsilon_{{\mathrm{p}}} = 10$ and {\color{black}$\epsilon_{{\mathrm{s}}} = 0.5$}. The maximum power radiated by each antenna is $P_{n} = 40 ~ \text{dBm}$, $\forall n$, the same as \cite{Chen2017}. The noise variance is set to $\sigma_{m}^{2} = 1$, $\forall m$. The threshold for the received SINR is $\gamma_{m} = 10 ~ \text{dB}$, $\forall m$. The number of constraints in Problem (\ref{prob_L21_F}) is {\color{black}$L = 38$}, the tuning parameter is\footnote{We do not study the relationship between $\eta$ and the sparsity of the solution, because of the space limitation of the paper. Instead, when we obtain ${\bf{\widehat{w}}}$, we choose $K$ antennas corresponding to the largest (in an $\ell_{2,1}$-norm sense) $K$ components. We have good results when the tuning parameter is $\eta = 0.1$.} $\eta = 0.1$, and the augmented Lagrangian parameter is $\rho = 50$ ($\rho/2 \gg \eta/L$ is satisfied). The value of ${\bf v}_{l}^{\text{(init)}}$ is given by any feasible point, while ${\bf u}_{l}^{\text{(init)}} = {\bf 0}$ and $k_{\mathrm{max}} = 100$. {\color{black}The beampatterns are drawn in Fig. \ref{Beampattern}, which indicates that both the proposed and FPP-SCA methods have beamlobes in the radar mainlobe $\mathcal{B}$ and the users directions. In addition, the proposed method has much higher response within $\mathcal{B}$ than the FPP-SCA method.}

\begin{figure}[t]
	\vspace*{-2mm}
	\centerline{\includegraphics[width=0.5\textwidth]{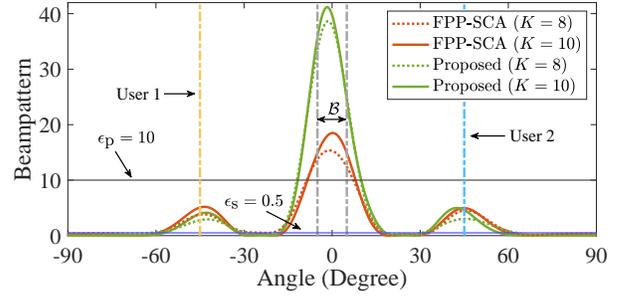}} 
	\caption{{\color{black}Beampattern comparison.}}
	\label{Beampattern}
\end{figure}

Next, we examine the TxPower and MSRR versus the number of selected sensors, i.e., $K$. We also consider a strategy of randomly selecting $K$ sensors. The parameters are the same as those in the last example. $500$ Monte-Carlo trials are performed. The results are plotted in Fig. \ref{vsNumSelectedSensor}. It is seen that our proposed method has the lowest transmit power and the highest MSRR, among all tested methods.

\begin{figure}[t]
	\vspace*{-2mm}
	\centerline{\includegraphics[width=0.5\textwidth]{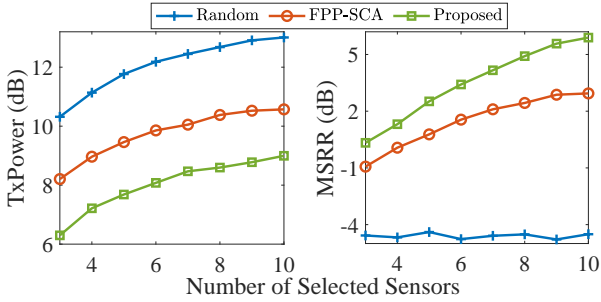}}  
	\caption{TxPower (left) and MSRR (right) versus $K$.}
	\label{vsNumSelectedSensor}
\end{figure}

Finally, we test the TxPower and MSRR versus the number of users, i.e., $M$. The number of selected antennas is fixed as $K = 8$, and the other parameters are unchanged as in the last example. The results are depicted in Fig. \ref{vsNumUser}, which again show that our algorithm leads to a more power-efficient solution.

\begin{figure}[t]
	\vspace*{-2mm}
	\centerline{\includegraphics[width=0.5\textwidth]{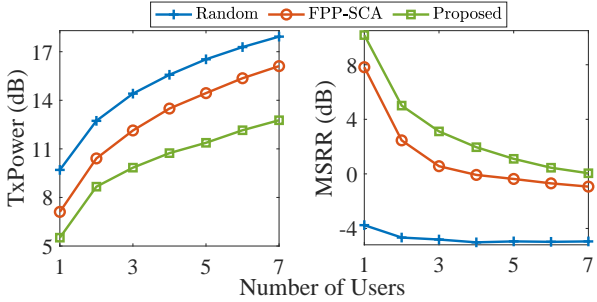}}  
	\caption{TxPower (left) and MSRR (right) versus $M$.}
	\label{vsNumUser}
\end{figure}

\section{Conclusion}
\label{conclusion}
We studied the problem of sparse array design for dual-function radar-communications. Our design aimed at maintaining good control in both mainlobe and sidelobe, and also to keep the signal-to-interference-plus-noise ratio for the users larger than a certain level. In addition, we considered the limitation of the radiated power by each antenna. The problem was formulated as a quadratically constrained quadratic program, and solved by the consensus alternating direction method of multipliers. The proposed algorithm was able to be implemented in parallel. Simulation results demonstrated better performance of the proposed algorithm than the other tested methods.

\appendices
\section{Proof of Theorem}    
\label{proof_v_closest}
Solving (\ref{v_l_hat}) is equal to finding a point in $\{\! {\bf v}_{l} \!:\! {\bf{ v}}_{l}^{\mathrm{H}}{\bf F}_{l}{\bf{ v}}_{l} \! \leq \! f_{l} \!\}$, such that it is closest (in an $\ell_{2}$-norm sense) to ${\bf{\bar v}}_{l}$. Hence, the theorem equivalently states that the solution to Problem (\ref{admm_v}) is the point closest (in an $\ell_{2}$-norm sense) to ${\bf{\bar v}}_{l}$, provided that ${\rho}/{2} \gg {\eta}/{L}$. To show this, we denote ${\bf{\widetilde{v}}}_{l}$ as the point in $\{{\bf v}_{l} : {\bf{ v}}_{l}^{\mathrm{H}}{\bf F}_{l}{\bf{ v}}_{l} \leq f_{l} \}$, such that 
\begin{align}
\label{v_widetilde}
\|{\bf{\widetilde{v}}}_{l} - {\bf{\bar{v}}}_{l}\|_{2} \leq \|{\bf v}_{l} - {\bf{\bar{v}}}_{l}\|_{2}
\end{align}
holds for any ${\bf v}_{l} \in \{{\bf v}_{l} : {\bf{ v}}_{l}^{\mathrm{H}}{\bf F}_{l}{\bf{ v}}_{l} \leq f_{l} \}$. Our goal is to show $f({\bf{\widetilde{v}}}_{l}) \leq f({\bf v}_{l})$, for any ${\bf v}_{l} \in \{{\bf v}_{l} : {\bf{ v}}_{l}^{\mathrm{H}}{\bf F}_{l}{\bf{ v}}_{l} \leq f_{l} \}$.

The Lagrangian parameter $\rho$ is chosen as $\rho = C \eta/L$, where $C$ is a constant. Then, $| g({\bf v}_{l}) - f({\bf v}_{l}) | \to 0$ as $C \to \infty$ (i.e., $\rho/2 \gg \eta/L$), where $ g({\bf v}_{l}) \triangleq \frac{\rho}{2}\|{\bf v}_{l} - {\bf w} + {\bf u}_{l}\|_{2}^{2}$. Moreover, 
\begin{align}
\label{f_v_approx}
f({\bf v}_{l}) = g({\bf v}_{l}) = \frac{\rho}{2}\|{\bf v}_{l} - {\bf{\bar{v}}}_{l}\|_{2}^{2},
\end{align}
as long as $C \to \infty$. Note that, in the second equality above, we used the fact that ${\bf{\bar{v}}}_{l} = {\bf w} - {\bf u}_{l}$ as $C \to \infty$. 

Suppose that there exists a point ${\bf{\breve v}}_{l} \in \{{\bf v}_{l} : {\bf{ v}}_{l}^{\mathrm{H}}{\bf F}_{l}{\bf{ v}}_{l} \leq f_{l} \}$, such that $f({\bf{\breve v}}_{l}) < f({\bf{\widetilde{v}}}_{l})$. Thus, by using (\ref{f_v_approx}), we obtain that $\frac{\rho}{2}\|{\bf{\breve v}}_{l} - {\bf{\bar{v}}}_{l}\|_{2}^{2} < \frac{\rho}{2}\|{\bf{\widetilde{v}}}_{l} - {\bf{\bar{v}}}_{l}\|_{2}^{2}$, which contradicts (\ref{v_widetilde}). This implies that $f({\bf{\widetilde{v}}}_{l}) \leq f({\bf v}_{l})$ holds for all feasible ${\bf v}_{l}$, that is, ${\bf{\widetilde{v}}}_{l}$ is the solution to Problem (\ref{admm_v}). This completes the proof.

%

\bibliographystyle{myIEEEtran}       
\bibliography{refs}

\end{document}